\documentclass[11pt,letterpaper]{article}

\usepackage{amsthm,color,latexsym,graphicx,url}
\usepackage[margin=1in]{geometry}
\usepackage{newpxtext,newpxmath}
\usepackage[utf8]{inputenc}
\usepackage{comment}
\usepackage{xcolor}
\usepackage{csquotes}\MakeOuterQuote{"}
\usepackage{gensymb}
\usepackage{subcaption}
\usepackage{overpic}

\usepackage{hyperref}
{\makeatletter \hypersetup{pdftitle={\@title}}}

{\makeatletter
	\gdef\xxxmark{%
		\expandafter\ifx\csname @mpargs\endcsname\relax 
		\expandafter\ifx\csname @captype\endcsname\relax 
		\marginpar{xxx}
		\else
		xxx 
		\fi
		\else
		xxx 
		\fi}
	\gdef\xxx{\@ifnextchar[\xxx@lab\xxx@nolab}
	\long\gdef\xxx@lab[#1]#2{\textbf{[\xxxmark #2 ---{\sc #1}]}}
	\long\gdef\xxx@nolab#1{\textbf{[\xxxmark #1]}}
}

{\makeatletter \gdef\fps@figure{!htbp}}

\let\realbfseries=\bfseries
\def\bfseries{\realbfseries\boldmath}

\newtheorem{theorem}{Theorem}[section]
\newtheorem{lemma}[theorem]{Lemma}

\newtheorem{corollary}[theorem]{Corollary}
\theoremstyle{definition}

\graphicspath{{./figures/}}

\let\epsilon=\varepsilon
\def\defn#1{\textbf{\textit{\boldmath #1}}}

\newcommand\manifold[1]{\mathcal{#1}}
\newcommand\graph[1]{\mathbf{#1}}
\newcommand\overlap{\cap}

\begin{document}

\title{All Polyhedral Manifolds are Connected by a 2-Step Refolding}

\author{
	Lily Chung\thanks{Computer Science and Artificial Intelligence Laboratory, Massachusetts Institute of Technology, USA.} 
	\and
	Erik D.~Demaine\footnotemark[1]
	\and
	Jenny Diomidova\footnotemark[1]
	\and
	Tonan Kamata\thanks{School of Information and Science, Japan
		Advanced Institute of Science and Technology, Japan.} 
	\and
	Jayson Lynch\footnotemark[1]
	\and
	Ryuhei Uehara\footnotemark[2]
	\and
	Hanyu Alice Zhang\thanks{School of Applied and Engineering Physics, Cornell University, USA.} 
}

\date{}

\maketitle

\begin{abstract}

We prove that, for any two polyhedral manifolds $\manifold P,\manifold Q$,
there is a polyhedral manifold $\manifold I$ such that
$\manifold P,\manifold I$ share a common unfolding and
$\manifold I,\manifold Q$ share a common unfolding.
In other words, we can unfold $\manifold P$,
refold (glue) that unfolding into $\manifold I$,
unfold $\manifold I$, and then refold into~$\manifold Q$.
Furthermore, if $\manifold P,\manifold Q$ have no boundary
and can be embedded in 3D (without self-intersection),
then so does~$\manifold I$.
These results generalize to $n$ given manifolds
$\manifold P_1,\manifold P_2, \dots, \manifold P_n$;
they all have a common unfolding with the same intermediate manifold~$\manifold I$.
Allowing more than two unfold/refold steps,
we obtain stronger results for two special cases:
for doubly covered convex planar polygons,
we achieve that all intermediate polyhedra are planar;
and for tree-shaped polycubes,
we achieve that all intermediate polyhedra are tree-shaped polycubes.

\end{abstract}

\section{Introduction}
\label{sec:intro}

Consider a \defn{polyhedral manifold} ---
a connected two-dimensional surface made from flat polygons
by gluing together paired portions of boundary
(but possibly still leaving some boundary unpaired,
and not necessarily embedded in space without overlap).
Two basic operations on such a manifold are \defn{gluing}
(joining together two equal-length portions of remaining boundary)
and the inverse operation \defn{cutting}
(splitting a curve into two equal-length portions of boundary,
while preserving overall connectivity of the manifold).
If we cut a manifold $\manifold P$ enough that it can be laid isometrically
into the plane (possibly with overlap), we call the resulting flat shape $U$
an \defn{unfolding} of~$\manifold P$.%
\footnote{Note that our notion of ``unfolding'' differs from many other uses,
  such as \cite{Demaine-O'Rourke-2007}, which forbid overlap.
  But it matches most previous work on refolding
  \cite{Refolding_CGTA,CommonUnfolding_CCCG2022}.}
Conversely, if we glue a flat shape $U$ into any polyhedral
manifold~$\manifold P$, we call $\manifold P$ a \defn{folding} of $U$
(and $U$ an unfolding of~$\manifold P$).

(Un)foldings naturally define an infinite bipartite graph $\graph G$
\cite[Section 25.8.3]{Demaine-O'Rourke-2007}:
define a vertex on one side for each manifold~$\manifold P$,
a vertex on the other side for each flat shape~$U$,
and an edge between $U$ and $\manifold P$ whenever
$U$ is an unfolding of~$\manifold P$
(or equivalently, $\manifold P$ is a folding of~$U$).
Because unfolding and folding preserve surface area,
we can naturally restrict the graph to manifolds and flat shapes of
a fixed surface area~$A$.
Is the resulting graph $\graph G_A$ connected?
In other words, is it possible to transform any polyhedral manifold
into any other polyhedral manifold of the same surface area
by an alternating sequence of unfolding to a flat shape,
folding that flat shape into a new manifold,
unfolding that manifold into a flat shape, and so on?
We call each pair of steps --- unfolding and then folding --- a
\defn{refolding step}.
We can then ask whether two manifolds have a $k$-step refolding
for each $k = 1, 2, \dots$.

In this paper, we give the first proof that the graph $\graph G_A$ is connected.
In fact, we show that the graph has diameter at most $2$:
every two polyhedral manifolds $\manifold P, \manifold Q$
have a $2$-step refolding.
In other words, there is a single polyhedral manifold $\manifold I$ such that
$\manifold P$ and $\manifold I$ share a common unfolding,
as do $\manifold I$ and $\manifold Q$.
This result turns out to follow relatively easily using classic results
from common dissection, similar in spirit to general algorithms for
hinged dissection \cite{HingedDissections_DCG}.

More interesting is that we show similar results when we restrict the
polyhedral manifolds to the following special cases, which sometimes reduce
the allowed input manifolds $\manifold P, \manifold Q$, but
importantly also reduce the allowed intermediate manifolds $\manifold I$:
\begin{enumerate}
\item \textbf{No boundary:}
  If polyhedral manifolds $\manifold P$ and $\manifold Q$ have no boundary
  (what we might call ``polyhedra''),
  then there is a 2-step refolding where the intermediate manifold
  $\manifold I$ also has no boundary.
  This version is similarly easy.
  (In fact, we can achieve this property even when $\manifold P$ and
  $\manifold Q$ have boundary.)
\item \textbf{Embedded:}
  If polyhedral manifolds $\manifold P$ and $\manifold Q$ are embedded in 3D
  and have no boundary,
  then there is a 2-step refolding where the intermediate manifold
  $\manifold I$ is embedded in 3D and has no boundary.
  This strengthening follows directly from a general theorem of
  Burago and Zalgaller \cite{Burago-Zalgaller-1996}.
  (In fact, we can achieve this property whenever $\manifold P$ and
  $\manifold Q$ are orientable.
  Or, if we allow $\manifold I$ to have boundary,
  we can always make it embeddable in 3D.)
\item \textbf{Doubly covered convex polygons:}
  If polyhedral manifolds $\manifold P$ and $\manifold Q$ are doubly covered
  convex polygons, then there is an $O(n)$-step refolding
  where every intermediate manifold $\manifold I$ is ``planar''
  (all polygons lie in the plane, but possibly with multiple layers)
  and has no boundary.
  This result follows from an $O(1)$-step refolding to remove a vertex from
  a doubly covered convex polygon.
\item \textbf{Polycubes:}
  If polyhedral manifolds $\manifold P$ and $\manifold Q$ are the surfaces
  of tree-shaped $n$-cubes (made from $n$ unit cubes joined face-to-face
  according to a tree dual), then there is an $O(n^2)$-step refolding
  where every intermediate manifold $\manifold I$ is a
  (possibly self-intersecting) tree-shaped $n$-cube.
  Furthermore, the refoldings involve cuts only along edges of the cubes
  (\defn{grid edges}); and if the given polycubes do not self-intersect and
  are ``slit-free'', then the intermediate polycubes also do not
  self-intersect.
  This result follows from simulating operations in reconfigurable robots.
\end{enumerate}

Past work on refolding has focused on the restriction to polyhedral manifolds
that are the surfaces of \emph{convex polyhedra}.
This version began with a specific still-open question ---
is there a 1-step refolding from a cube to a regular tetrahedron? ---
independently posed by M. Demaine (1998), F. Hurtado (2000), and E. Pegg (2000).
When E. Demaine and J. O'Rourke wrote this problem in their book
\cite[Open Problem 25.6]{Demaine-O'Rourke-2007},
they also introduced the multi-step refolding problem.
Let $\graph C_A$ be the subgraph of $\graph G_A$
restricting to convex polyhedra and their unfoldings.
Demaine, Demaine, Diomidova, Kamata, Uehara, and Zhang
\cite{Refolding_CGTA} showed that several convex polyhedra of surface area $A$
are all in the same connected component of~$\graph C_A$:
doubly covered triangles,
doubly covered regular polygons,
tetramonohedra (tetrahedra whose four faces are congruent acute triangles, including doubly covered rectangles),
regular prisms,
regular prismatoids,
augmented regular prismatoids,
and all five Platonic solids.
These refoldings require just $O(1)$ steps (at most~$9$).

Our 2-step refolding is very general,
applying in particular to any two convex polyhedra;
for example, Figure~\ref{fig:cube-tetrahedron} shows the
example of a cube to a regular tetrahedron.
But our refolding crucially relies on a
nonconvex intermediate manifold~$\manifold I$.
We conjecture that two steps is also optimal, even for two convex polyhedra.
Indeed, Arseneva, Demaine, Kamata, and Uehara \cite{CommonUnfolding_CCCG2022}
conjectured that most pairs of doubly covered triangles
(specifically, those with rationally independent angles)
have no common unfolding, and thus no 1-step refolding.
As evidence, they showed (by exhaustive search) that any common unfolding
has at least $300$ vertices.
Assuming this conjecture, two steps are sometimes necessary.
Two steps is also the first situation where we have an intermediate manifold
$\manifold I$, which is what allows us to exploit the additional freedom
of the nonconvexity of~$\manifold I$.

After presenting our general 2-step refolding (Section~\ref{sec:polyhedral-manifolds-transformation}),
we consider the special cases of doubly covered convex polygons (Section~\ref{sec:doubly-covered-shapes-transformation}) and polycubes (Section~\ref{sec:polycube-transformation}).

\section{Refolding Model}

%

In our constructions, we use a more general but equivalent form of
``refolding step'': any cutting followed by any gluing.
In other words, we allow using an arbitrary connected manifold
in between cutting and gluing.
By contrast, the definition in Section~\ref{sec:intro}
requires a full unfolding followed by a folding,
which requires a flat shape in between cutting and gluing.

These two models are equivalent.
If we want to modify one of our refolding steps to instead reach a full
unfolding after cutting, we can perform additional cuts (that preserve
connectivity) until the manifold can be laid flat,
and then immediately reglue those cuts back together.
The same idea is used in \cite{Refolding_CGTA}.

\section{Transformation Between Polyhedral Manifolds}
\label{sec:polyhedral-manifolds-transformation}

In this section, we prove the main result of the paper:

\begin{theorem} \label{thm:polyhedral-manifolds}
  For any $n$ polyhedral manifolds $\manifold P_1, \dots, \manifold P_n$
  of the same surface area, there is another polyhedral manifold $\manifold I$
  such that $\manifold P_i$ and $\manifold I$ have a common unfolding
  for all~$i$.
  We can guarantee that the intermediate manifold $\manifold I$ has no boundary,
  or guarantee that it embeds in 3D.
  If manifolds $\manifold P_i$ are all orientable,
  then we can guarantee that the intermediate manifold $\manifold I$
  is orientable, has no boundary, and embeds in 3D.
\end{theorem}

We initially focus on the case of $n=2$ manifolds:

\begin{corollary} \label{cor:polyhedral-manifolds}
  Any two polyhedral manifolds $\manifold P, \manifold Q$
  of the same surface area have a 2-step refolding.
  We can guarantee that the intermediate manifold $\manifold I$ has no boundary,
  or guarantee that it embeds in 3D.
  If manifolds $\manifold P$ and $\manifold Q$ are both orientable,
  then we can guarantee that the intermediate manifold $\manifold I$
  is orientable, has no boundary, and embeds in 3D.
\end{corollary}

Figure~\ref{fig:cube-tetrahedron} gives an example construction of an
intermediate manifold $\manifold I$ when $\manifold P$ is a cube and
$\manifold Q$ is a regular tetrahedron.
To improve figure clarity, this construction does not exactly follow
the general algorithm described below:
instead of the general dissection algorithm, we use an efficient 3-piece
dissection based on a dissection of Gavin Theobald \cite{latin-cross},
and we preserve more original gluings when resolving overlaps.
Nonetheless, it serves as a running example of the key steps in our algorithm.

\begin{figure}
  \centering
  \subcaptionbox{\label{fig:cube-tetrahedron-dissection}
    3-piece dissection of the cube into the regular tetrahedron,
    based on Gavin Theobald's 5-piece dissection of the Latin cross
    into the equilateral triangle \cite{latin-cross}. Red and blue dashed lines represent the folding to a cube and regular tetrahedron respectively.}
  {\includegraphics[page=1,width = \linewidth]{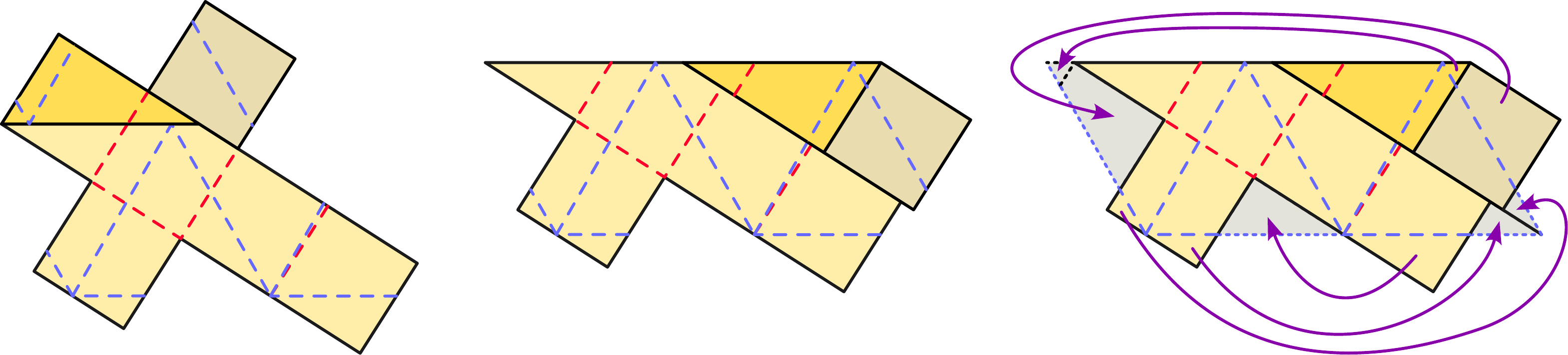}}

  \subcaptionbox{\label{fig:cube-tetrahedron-desired}
    Desired gluings for the cube (red) and the regular tetrahedron (blue).}
  {\includegraphics[page=2,width = 0.475\linewidth]{3piece_cube_tetra.pdf}}
  \hfill
  \subcaptionbox{\label{fig:cube-tetrahedron-final}
    Gluing for an intermediate manifold $\manifold I$,
    from subsets of the cube (red) and the tetrahedron (blue) gluing
    and zipping the remainder (green).}
  {\includegraphics[page=3,width = 0.475\linewidth]{3piece_cube_tetra.pdf}}

  \caption{Example 2-step refolding from the cube to the regular tetrahedron.}
  \label{fig:cube-tetrahedron}
\end{figure}

\subsection{Common Dissection}
\label{sec:common-dissection}

We start by computing a \defn{common dissection} of the given manifolds
$\manifold P$ and $\manifold Q$ of equal surface area, that is,
a subdivision of each surface into polygons that match in
the sense that, for some perfect pairing of $\manifold P$'s polygons
with $\manifold Q$'s polygons, there is an isometry between paired polygons.
Solutions to this dissection problem for \emph{polygons} $\manifold P$ and
$\manifold Q$ go back to the early 1800s
\cite{lowry1814solution,wallace1831elements,bolyai1832tentamen,gerwien1833zerschneidung}.
Their high-level approach is as follows:
\begin{enumerate}
\item Triangulate $\manifold P$ and $\manifold Q$.
\item Dissect each triangle (from both triangulations) into a rectangle.
  (This dissection needs only three pieces: cut the triangle
  parallel to its base at half the height, and cut from the apex
  orthogonal to the first cut.)
\item Dissect each rectangle into a rectangle of height $h_{\min}$,
  the smallest height among all the rectangles
  (i.e., half the smallest height among all triangles in both triangulations).
  (This dissection is more difficult and requires a pseudopolynomial number of pieces.)
\item Arrange all the rectangles from $\manifold P$ into one long rectangle
  of height $h_{\min}$, and similarly arrange all the rectangles from
  $\manifold Q$ into one long rectangle of height $h_{\min}$,
  necessarily the same rectangle~$R$.
\item Overlay the two dissections of this common rectangle $R$ and subdivide
  according to all cuts, producing a set of polygons
  that can form into $\manifold P$ and can form into~$\manifold Q$.
\end{enumerate}
See \cite{FeatureSize_EGC2011f,HingedDissections_DCG}
for more algorithmic descriptions,
including pseudopolynomial bounds on the number of pieces.

We can apply the same technique to the case where
$\manifold P$ and $\manifold Q$ are polyhedral manifolds instead of polygons.
The only slightly different step is triangulating the surfaces $\manifold P$
and $\manifold Q$ (Step~1), which we can do by e.g.\ triangulating the faces.
Then Steps 2--5 apply to the resulting triangles as usual.
Because the dissection construction does not require flipping the polygons,
the resulting common dissection is \defn{locally orientation preserving}:
the mapping from the polygons arranged to form $\manifold P$
to the polygons arranged to form $\manifold Q$
locally preserves which side of each polygon is ``up''.

The resulting dissection may not be ``edge-to-edge'':
when assembling the polygons together to form $\manifold P$ or $\manifold Q$,
two polygons may meet (intersect) at a segment that is only a subset of
an edge of either polygon.
Figure~\ref{fig:cube-tetrahedron-dissection} shows an example of such a
dissection, from a cube to a regular tetrahedron:
for example, in the cube (cross) arrangement, the triangular piece and
square piece share only a portion of their edges.
(This example was designed by hand, based on a dissection by Gavin Theobald
\cite{latin-cross}, not produced by the algorithm above.
In fact, Gavin Theobald found a 2-piece dissection from the cube to
the regular tetrahedron \cite{cube-tetra-2}.)

We can generalize this common dissection construction to
$n$ polyhedral manifolds $\manifold P_1, \dots, \manifold P_n$
of the same surface area
(as also mentioned in \cite{HingedDissections_DCG}):
just overlay $n$ dissections in Step~5.
Henceforth we will consider the case of general~$n$.

\subsection{Abstract Intermediate Manifold}
\label{sec:abstract-transformation}

Next we construct the intermediate manifold $\manifold I$.
For now, we will not worry about embeddability,
and just construct an (abstract) polyhedral manifold.

Consider two polygons $P_1, P_2$ in the common dissection
that are adjacent in manifold $\manifold P_i$
meaning that, when the polygons are assembled to form $\manifold P_i$,
there is an edge $e_1$ of $P_1$ and an edge $e_2$ of $P_2$
that overlap on a common positive-length segment.
Let $\overlap_i(e_1, e_2)$ denote the segment of $e_1$
that intersects $e_2$ when assembling $\manifold P_i$, and let
$\overlap_i(e_2, e_1)$ denote the corresponding segment of $e_2$.
(If the edges share more than one segment,
as they might in a non-edge-to-edge gluing,
pick one arbitrarily,
but consistently for $e_1$ and~$e_2$.)
Intuitively,
$\overlap_i(e_1, e_2) \leftrightarrow
\overlap_i(e_2, e_1)$
represent the gluings desired by $\manifold P_i$,
but the $\manifold P_i$ gluings likely conflict with the $\manifold P_j$
gluings for $i \neq j$.
Figure~\ref{fig:cube-tetrahedron-desired} gives an example.

Next we construct a (partial) gluing on the boundaries of the polygons
that includes a positive segment from every
$\overlap_i(e_1, e_2) \leftrightarrow
\overlap_i(e_2, e_1)$ gluing,
while avoiding conflicts.
The algorithm proceeds as follows:
\begin{enumerate}
\item For each manifold $\mathcal P_i$,
  and for every overlapping pair of edges $e_1$ of $P_1$ and
  $e_2$ of $P_2$ when assembling the polygons into $\mathcal P_i$,
  add $\overlap_i(e_1, e_2) \leftrightarrow \overlap_i(e_2, e_1)$
  to the list of gluings.
\item
  \label{step:rep:first}
  Find two gluings that overlap on an edge $e$, say
  $s_1 \leftrightarrow s_2$ and
  $s'_1 \leftrightarrow s'_2$
  where segments $s_1$ and $s'_1$ are subsegments of a common edge $e$
  that overlap on a positive-length segment $s_1 \cap s'_1$.
  Refer to Figure~\ref{fig:conflict resolve}.

\begin{figure}
  \centering
  \includegraphics[scale=0.2]{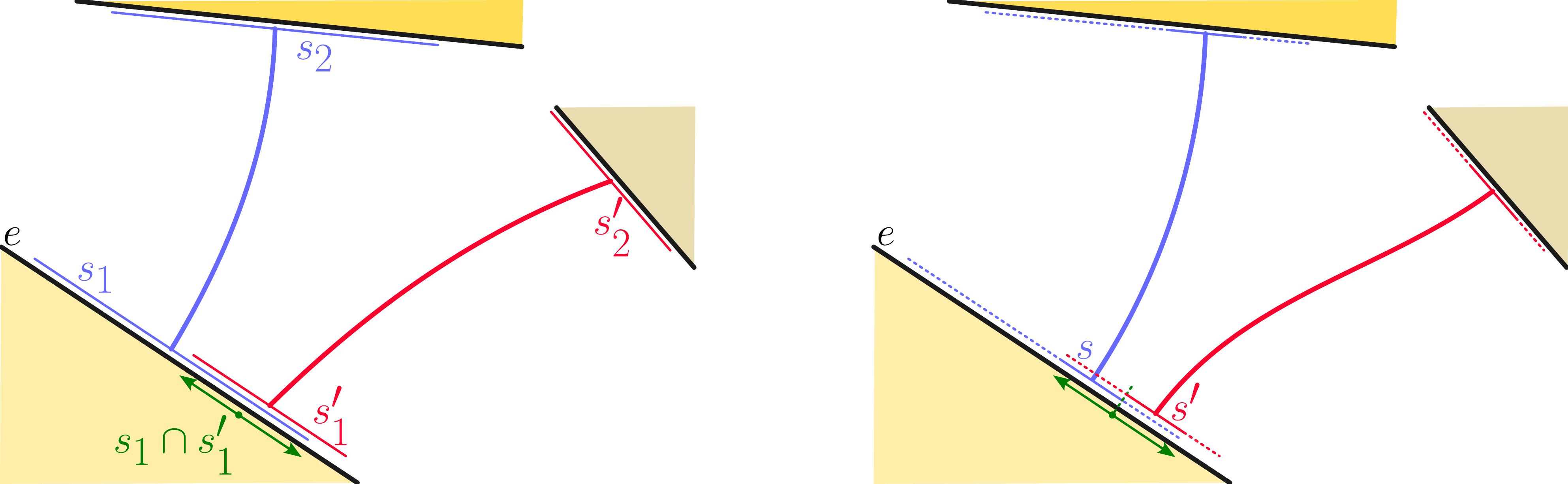}
  \caption{Removing overlap from two gluings $s_1 \leftrightarrow s_2$ and
    $s'_1 \leftrightarrow s'_2$ on edge $e$,
    by reducing $s_1$ and $s'_1$ to subsegments $s$ and $s'$
    which bisect $s_1 \cap s'_1$.}
  \label{fig:conflict resolve}
\end{figure}

\item
  Divide the segment $s_1 \cap s'_1$ into two equal halves,
  $s$ and~$s'$.
\item
  \label{step:rep:last}
  Remove the overlap between these two gluings by restricting
  $s_1 \leftrightarrow s_2$ to the subsegment $s \subset s_1$,
  and restricting
  $s'_1 \leftrightarrow s'_2$ to the subsegment $s' \subset s'_1$.
\item
  Repeat Steps~\ref{step:rep:first}--\ref{step:rep:last}
  until all overlaps have been removed.
\end{enumerate}
Because this algorithm only modifies gluings by restricting to a subsegment,
it never adds new overlaps,
so it will remove all overlaps after $O((n \, E^2)^2)$ repetitions,
where $E$ is the number of edges in the common dissection
so $O(E^2)$ bounds the number of gluings from each manifold $\mathcal P_i$.
Furthermore, every original gluing
$\overlap_i(e_1, e_2) \leftrightarrow \overlap_i(e_2, e_1)$
remains intact for some positive length.

Some of the boundary of the polygons may now be unglued.
If we want to avoid $\manifold I$ having boundary,
we can glue each segment $s$ of remaining boundary to itself,
by dividing it in two equal halves $s', s''$,
giving $s'$ and $s''$ opposite orientations,
and gluing $s' \leftrightarrow s''$.
(This type of gluing is called ``zipping'' \cite{Aleks_GC2002}.)
Figure~\ref{fig:cube-tetrahedron-final} gives an example of
a gluing that might result from an optimized form of this algorithm
(where we maintain as much of the original gluings as possible).

The gluing described above defines the intermediate manifold $\manifold I$.
Because manifold $\manifold I$ contains a portion of every desired gluing for
$\manifold P_i$, $\manifold I$ has a common unfolding with $\manifold P_i$:
just cut all gluings that did not originate from~$\manifold P_i$.
Because no cut fully separates an entire edge from its mate in $\manifold P_i$,
this cutting preserves connectivity; indeed, we obtain the same dual graph
of piece adjacencies as we do when arranging the dissection into~$\manifold P_i$.
Thus we obtain a 1-step refolding from $\manifold P_i$ to $\manifold I$,
and similarly from $\manifold I$ to $\manifold P_j$,
which gives a 2-step refolding from $\manifold P_i$ to $\manifold P_j$.

If manifolds $\manifold P_i$ are all orientable,
then so is the resulting intermediate manifold $\manifold I$,
because the common dissection is (locally) orientation preserving.

\subsection{Embeddable Intermediate Polyhedron via Burago--Zalgaller Theorem}
\label{sec:embeddable-transformation}

To guarantee that $\manifold I$ is embeddable in 3D,
we use a powerful result of Burago and Zalgaller \cite{Burago-Zalgaller-1996}.
See also \cite{O'Rourke-2010} for a detailed description of the result, and
\cite{Saucan-2012} for a description of the (quite complicated) construction.

\begin{theorem}[{\cite[Theorem~1.7]{Burago-Zalgaller-1996}}] \label{thm:B-Z}
  Every polyhedral manifold that is either orientable or has boundary
  admits an isometric piecewise-linear $C^0$ embedding into 3D.
\end{theorem}

To apply this theorem, we need that $\manifold I$ is either
orientable or has boundary.
As argued above, $\manifold I$ is orientable
if manifolds $\manifold P_i$ are all orientable.
Otherwise, we can give $\manifold I$ boundary by reducing any one gluing
to half of its length (or omitting the zips if we had some).
Either way, Theorem~\ref{thm:B-Z}
gives a subdivision of the polygons in $\manifold I$
into finitely many subpolygons, each of which gets isometrically embedded in 3D
by an embedding of~$\manifold I$.

%
%

\section{Transformation Between Doubly Covered Convex Polygons}
\label{sec:doubly-covered-shapes-transformation}

We start with \defn{doubly covered convex polygons}, that is,
polyhedral manifolds without boundary formed from two copies of a
convex planar polygon by gluing together all corresponding pairs of edges.
Here we require that every intermediate polyhedral manifold $\manifold I$
is \defn{planar} in the sense that its polygons all lie in the plane,
but we allow any number of layers of stacked polygons,
generalizing the notion of doubly covered polygon.

\begin{theorem}\label{thm:convex-polygon-refolding}
  Any two doubly covered convex $n$-gons of the same area
  have an $O(n)$-step refolding, where all intermediate manifolds are planar with no boundary.
\end{theorem}

As a useful building block, we consider a simple 2-step refolding which allows removing a piece from the polygon, rotating it, and gluing it back elsewhere (similar to hinged dissection), provided we can fold the polygon to facilitate the gluing.

\begin{lemma}\label{lem:rearrange}
  Let $P$ be a subset of the plane homeomorphic to a closed disk,
  as visualized in Figure~\ref{fig:rearrange}.
  Suppose $A_1B_1$ and $A_2B_2$ are line segments on the boundary of $P$, such that there exists a plane reflection $r$ taking $A_1$ to $A_2$ and $B_1$ to $B_2$.
  Let $c$ be a simple curve starting and ending on the boundary of $P$ passing through its interior, so that it separates $P$ into two closed halves $P_1$ and $P_2$ containing $A_1B_1$ and $A_2B_2$ respectively.
  Let $f$ be the unique rotation and translation such that $f(A_1) = A_2$ and $f(B_1) = B_2$, and suppose $f(P_1)$ intersects $P_2$ only on $A_2B_2$.
  Then there is a 2-step refolding between the double covers of $P$ and $P'$, where $P' = f(P_1) \cup P_2$.
\end{lemma}
\begin{proof}
  The refolding is accomplished by the following steps,
  illustrated in Figure~\ref{fig:rearrange}:
  \begin{enumerate}
  \item Cut along $A_1B_1$ and $A_2B_2$, then glue the top layer of $A_1B_1$ to the top layer of $A_2B_2$ and similarly for the bottom layers.  This intermediate step can be folded flat by a single fold along the line of reflection of $r$.
  \item Cut along $c$ in both layers to create two new boundaries $c_1$ and $c_2$.  Then glue the top layer of $c_1$ to the bottom layer of $c_1$ and similarly for $c_2$.
  \qedhere
  \end{enumerate}
\end{proof}

\begin{figure}
  \centering
  \subcaptionbox{Folding to glue $A_1 B_1$ to $A_2 B_2$}{
    \begin{overpic}[scale=.2]{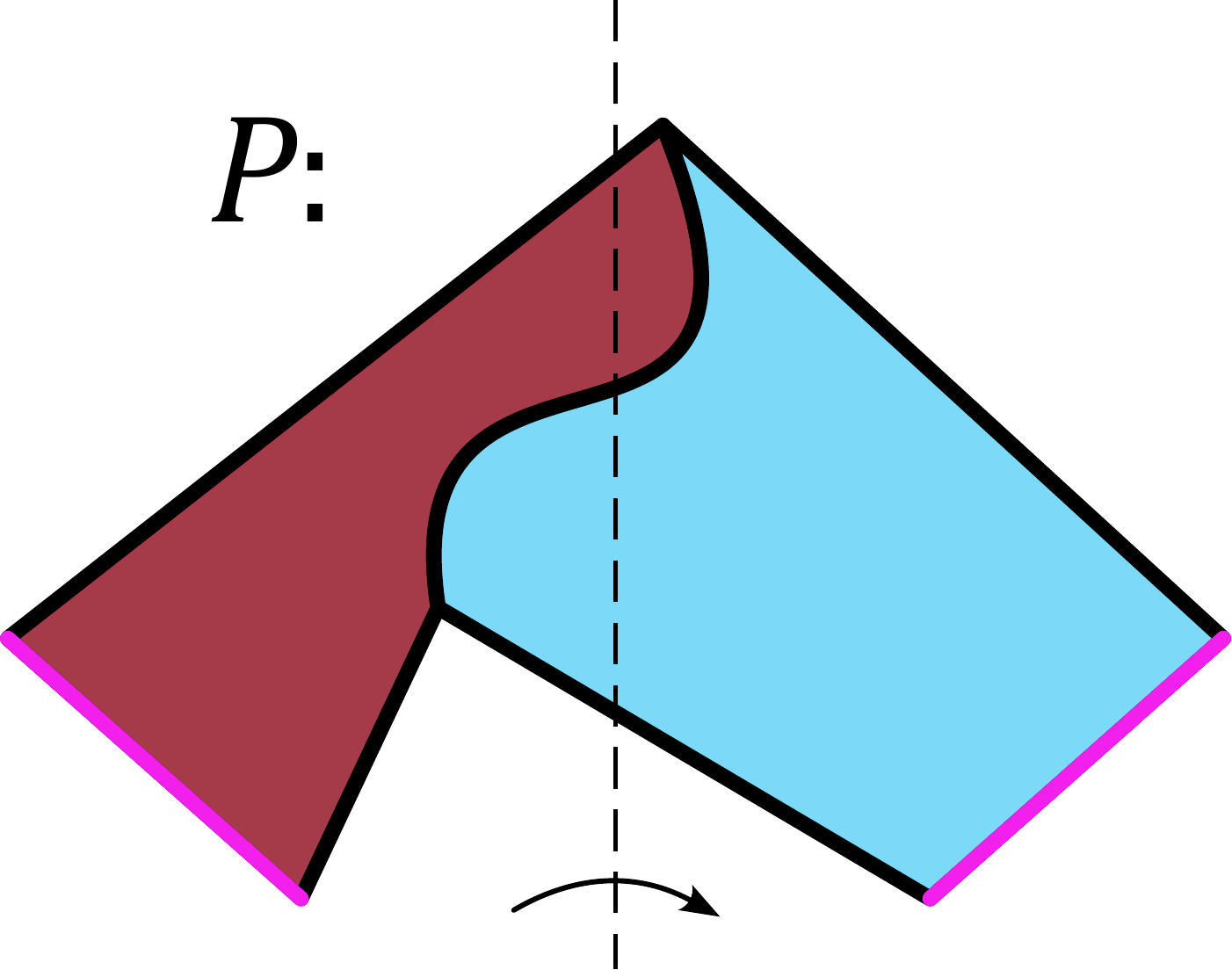}
      \put(15,-2){$A_1$}
      \put(-9,22){$B_1$}
      \put(75,-2){$A_2$}
      \put(102,22){$B_2$}
    \end{overpic}
    \vspace{1em}
  }
  \hfill
  \subcaptionbox{Cutting along $c$}[3cm]{
    \begin{overpic}[scale=.2]{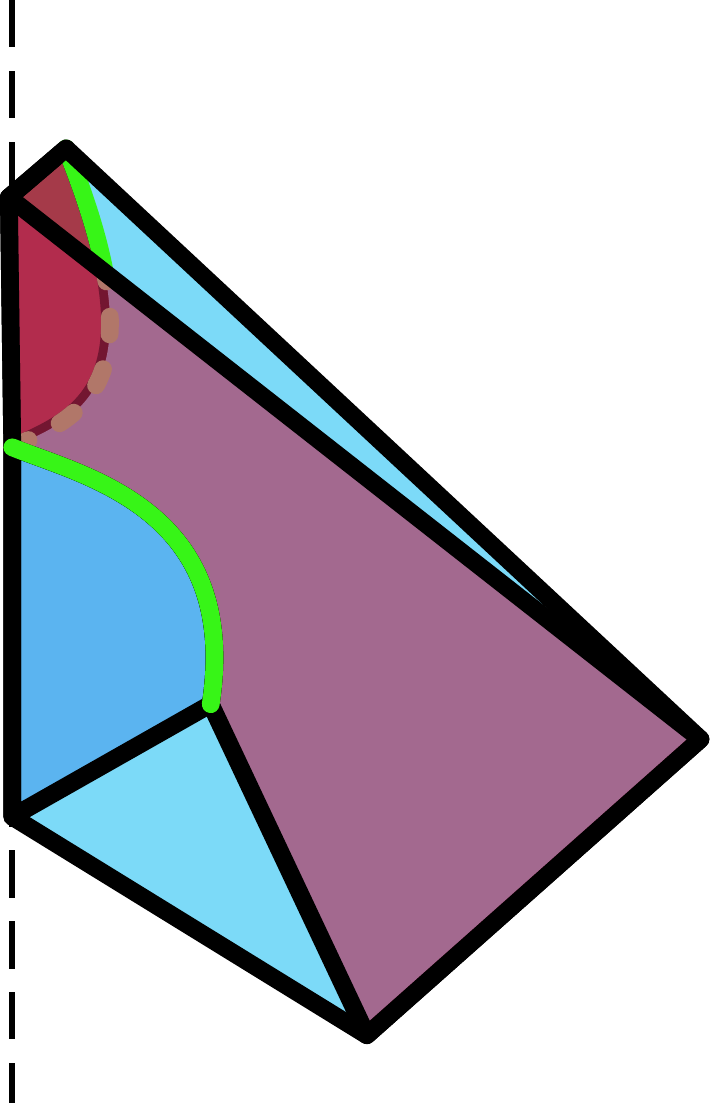}
      \put(34,-2){$A_2$}
      \put(65,22){$B_2$}
    \end{overpic}
    \vspace{1em}
  }
  \hfill
  \subcaptionbox{Finished}{
    \begin{overpic}[scale=.2]{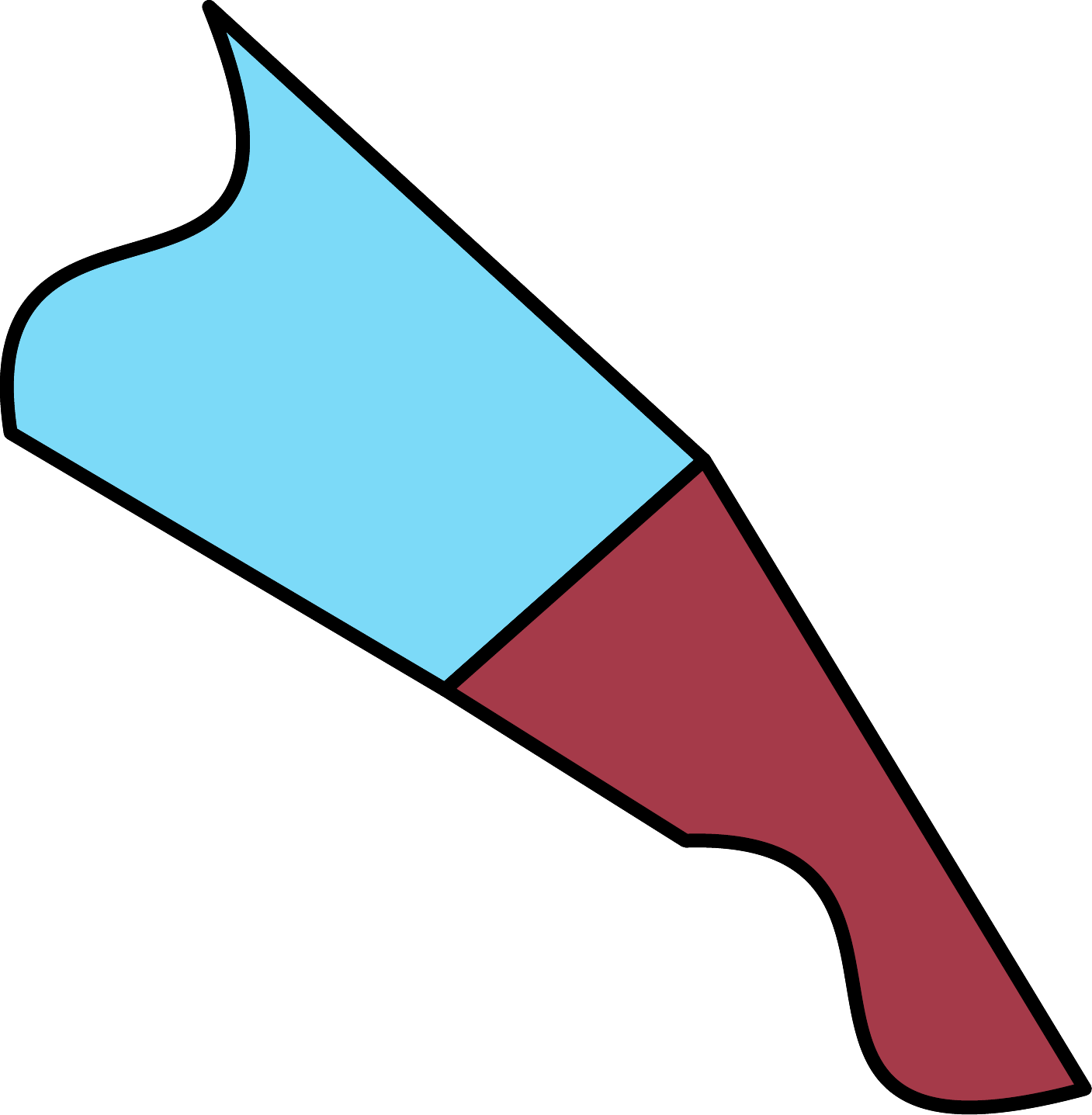}
      \put(30,28){$A_2$}
      \put(65,62){$B_2$}
    \end{overpic}
  }
  \caption{Rearranging two pieces via a 2-step refolding.}
  \label{fig:rearrange}
\end{figure}

Now we consider the triangle $\triangle ABC$ formed by three consecutive vertices $A,B,C$ on the boundary of a polygon $P$.
Our goal is to find an $O(1)$-step refolding of $P$ which moves the apex $B$ parallel to $AC$ (which preserves area).
This will allow us to move $B$ so that the interior angle at $C$ becomes $180\degree$, eliminating a vertex from~$P$.
By induction, this allows us to reduce any doubly covered polygon down
to a doubly covered triangle, and then we can use
a known 3-step refolding between doubly covered triangles
\cite[Theorem~2]{Refolding_CGTA}.
We accomplish the goal as follows:

\begin{lemma}\label{lem:acute}
  Let $P$ be a convex polygon with three consecutive vertices $A, B, C$
  such that the projection of $B$ onto $AC$ is between $A$ and $C$.
  Then there is an $O(1)$-step refolding between the double covers of $P$ and $P'$, where $P'$ is the polygon obtained from $P$ by replacing $\triangle ABC$ by a rectangle with base $AC$ with the same area.
\end{lemma}
\begin{proof}
  Refer to Figure~\ref{fig:acute}.
  Let $X$ be the midpoint of $AB$, $Y$ be the midpoint of $BC$, and $O$ be the projection of $B$ onto $XY$.  Using Lemma~\ref{lem:rearrange}, we rotate $\triangle XBO$ by $180\degree$ about $X$,
  and similarly we rotate $\triangle YBO$ by $180\degree$ about $Y$.
  This forms the desired rectangle.
\end{proof}

\begin{figure}
  \centering
  \begin{overpic}[scale=.3]{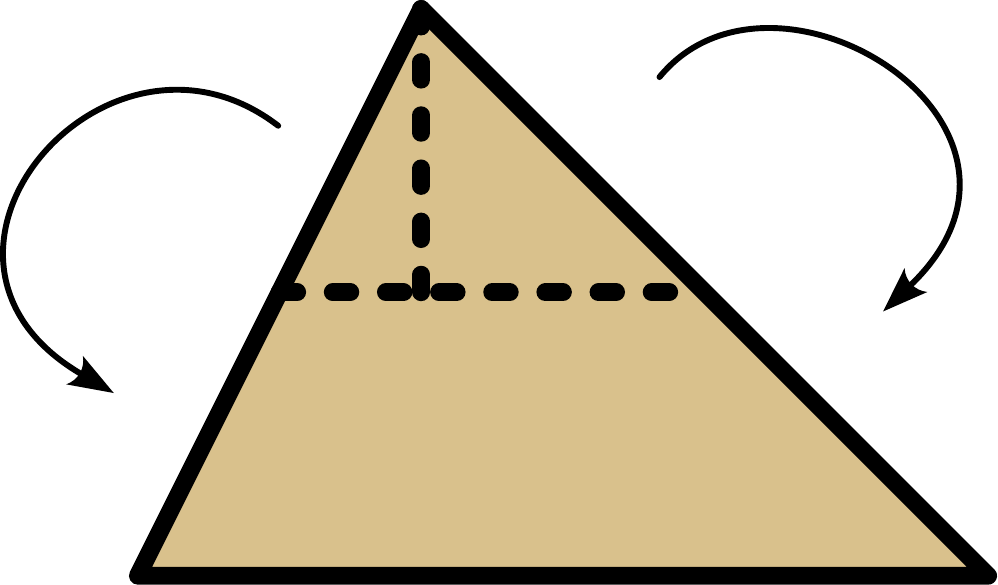}
    \put(6,-2){\makebox(0,0)[c]{\strut $A$}}
    \put(42,65){\makebox(0,0)[c]{\strut $B$}}
    \put(106,-2){\makebox(0,0)[c]{\strut $C$}}
    \put(22,32){\makebox(0,0)[c]{\strut $X$}}
    \put(78,32){\makebox(0,0)[c]{\strut $Y$}}
    \put(42,22){\makebox(0,0)[c]{\strut $O$}}
  \end{overpic}
  \caption{Refolding $\triangle ABC$ into a rectangle when $B$ is between $A$ and $C$.}
  \label{fig:acute}
\end{figure}

\begin{lemma}\label{lem:scalene}
  Let $P$ be a convex polygon with three consecutive vertices $A, B, C$,
  and let $B_1$ be the projection of $B$ onto $AC$.
  Suppose that $C$ is between $A$ and $B_1$,
  and $|CB_1| \le 4|AC|$.
  Then there is an $O(1)$-step refolding between the double covers of $P$ and $P'$, where $P'$ is the polygon obtained from $P$ by replacing $\triangle ABC$ by a rectangle with base $AC$ with the same area.
\end{lemma}
\begin{proof}
  Refer to Figure~\ref{fig:scalene}, where all point labels remain fixed in the plane across all subfigures.
  Let $X$ be the midpoint of $AB$ and $Y$ be the midpoint of $BC$
  (Figure~\ref{fig:scalene-1}).
  For this proof, we will adopt the convention that $p_1$ denotes the projection of point $p$ onto $AC$ and $p_2$ denotes the projection of $p$ onto $XY$.

  Using Lemma~\ref{lem:rearrange}, we rotate $\triangle XYB$ by $180\degree$ about $X$ (Figure~\ref{fig:scalene-1}), forming the parallelogram $ACYY'$ (Figure~\ref{fig:scalene-2}).
  Now let $W$ be the midpoint of $CY$ and $V$ be the midpoint of $AY'$.
  We have
  \[|AV_1| = |W_2Y| = \frac14|CB_1| \le |AC| = |Y'Y|,\]
  which implies $V_1$ lies on $AC$ and $W_2$ lies on $Y'Y$.
  Using Lemma~\ref{lem:rearrange} twice, we rotate $\triangle AVV_1$ by $180\degree$ about $V$, and $\triangle YWW_2$ by $180\degree$ about $W$.
  This forms a rectangle $V_1 V_2 W_2 W_1$.
  Finally, we use Lemma~\ref{lem:rearrange} again to move the rectangle $C C_2 W_2 W_1$ to $A A_2 V_1 V_2$.
\end{proof}

\begin{figure}
  \centering
  \subcaptionbox{\label{fig:scalene-1}Triangle to parallelogram}{
    \begin{overpic}[scale=.15]{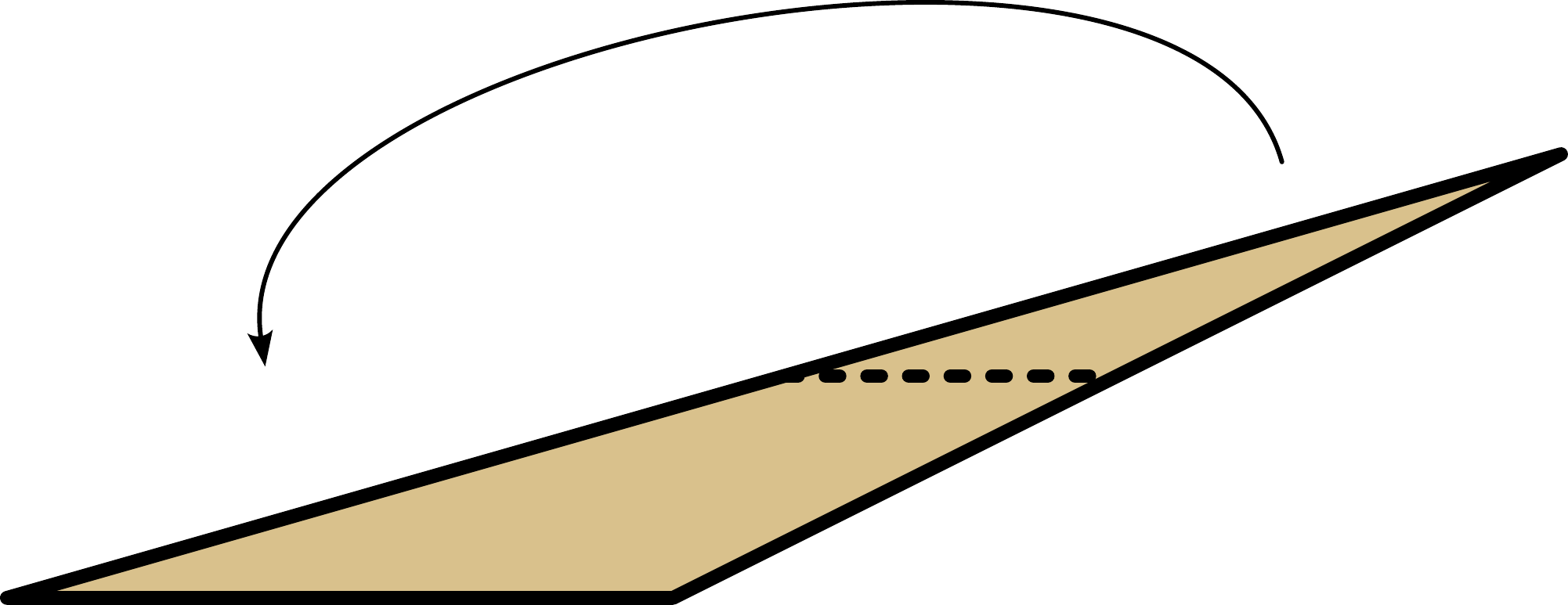}
    \put(-2,-4){\makebox(0,0)[c]{\strut $A$}}
    \put(104,30){\makebox(0,0)[c]{\strut $B$}}
    \put(45,-4){\makebox(0,0)[c]{\strut $C$}}
    \put(47,18){\makebox(0,0)[c]{\strut $X$}}
    \put(75,10){\makebox(0,0)[c]{\strut $Y$}}
    \end{overpic}
    \vspace{1em}
  }
  \hfill
  \subcaptionbox{\label{fig:scalene-2}Parallelogram to rectangle}[5cm]{
    \begin{overpic}[scale=.15]{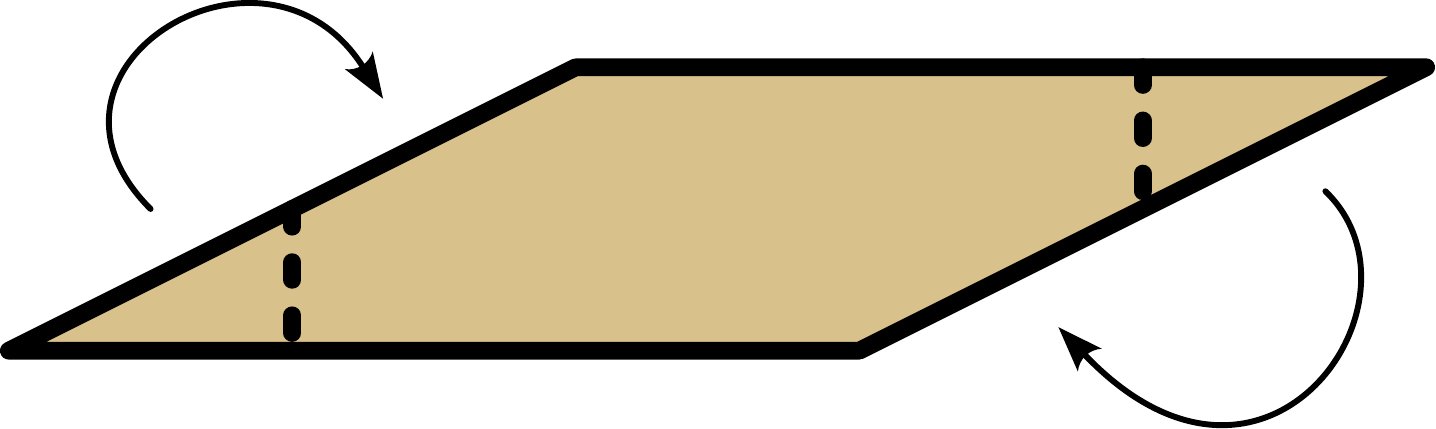}
    \put(-2,-2){\makebox(0,0)[c]{\strut $A$}}
    \put(104,30){\makebox(0,0)[c]{\strut $Y$}}
    \put(60,-2){\makebox(0,0)[c]{\strut $C$}}
    \put(40,30){\makebox(0,0)[c]{\strut $Y'$}}
    \put(85,8){\makebox(0,0)[c]{\strut $W$}}
    \put(80,32){\makebox(0,0)[c]{\strut $W_2$}}
    \put(18,20){\makebox(0,0)[c]{\strut $V$}}
    \put(22,-2){\makebox(0,0)[c]{\strut $V_1$}}
    \end{overpic}
    \vspace{1em}
  }
  \hfill
  \subcaptionbox{\label{fig:scalene-3}Shifting the rectangle}{
    \begin{overpic}[scale=.15]{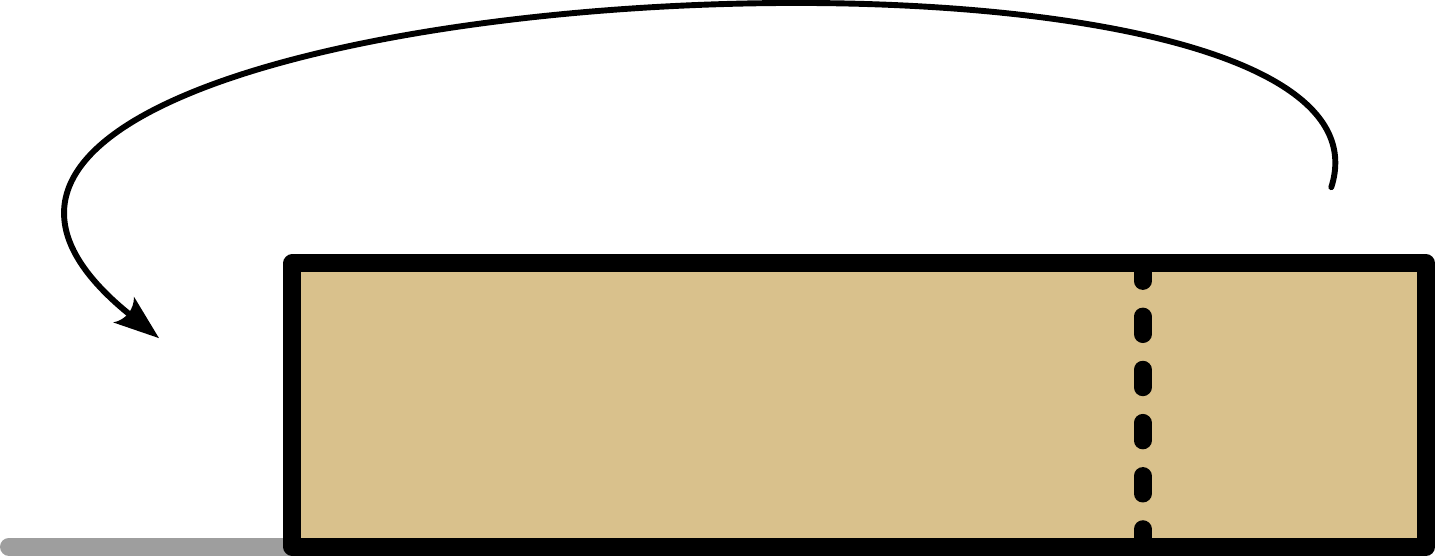}
    \put(-2,-7){\makebox(0,0)[c]{\strut $A$}}
    \put(-5,25){\makebox(0,0)[c]{\strut $A_2$}}
    \put(80,27){\makebox(0,0)[c]{\strut $C_2$}}
    \put(106,25){\makebox(0,0)[c]{\strut $W_2$}}
    \put(106,-7){\makebox(0,0)[c]{\strut $W_1$}}
    \put(80,-7){\makebox(0,0)[c]{\strut $C$}}
    \put(22,-7){\makebox(0,0)[c]{\strut $V_1$}}
    \put(22,27){\makebox(0,0)[c]{\strut $V_2$}}
    \end{overpic}
    \vspace{1em}
  }
  \caption{Refolding $\triangle ABC$ into a rectangle when $B$ is not between $A$ and $C$.}
  \label{fig:scalene}
\end{figure}

\begin{lemma}\label{lem:convex-projection}
  Let $P$ be a convex polygon with three consecutive vertices $A, B, C$
  and let $Z$ be another vertex of $P$ such that the interior angle at $Z$ is at most the interior angle at $B$.
  Let $\ell$ be the line through $B$ parallel to $AC$, $Q_1$ be the intersection of $ZA$ with $\ell$, and $Q_2$ be the intersection of $ZC$ with $\ell$.
  Then $\min\{|Q_1B|, |Q_2B|\} \le |AC|$.
\end{lemma}
\begin{proof}
  Refer to Figure~\ref{fig:convex-projection}.
  Construct $B'$ so that $ABCB'$ is a parallelogram.
  Vertex $Z$ cannot lie in the interior of $\triangle ACB'$ or else its interior angle would be larger than that of $B$ (by convexity of~$P$).
  Thus $Z$ is either below $AB'$ or below $CB'$; 
  in the first case, we have $|Q_1B| \le |AC|$, and in the second case, we have $|Q_2B| \le |AC|$.
\end{proof}

\begin{figure}
  \centering
  \begin{overpic}[scale=.3]{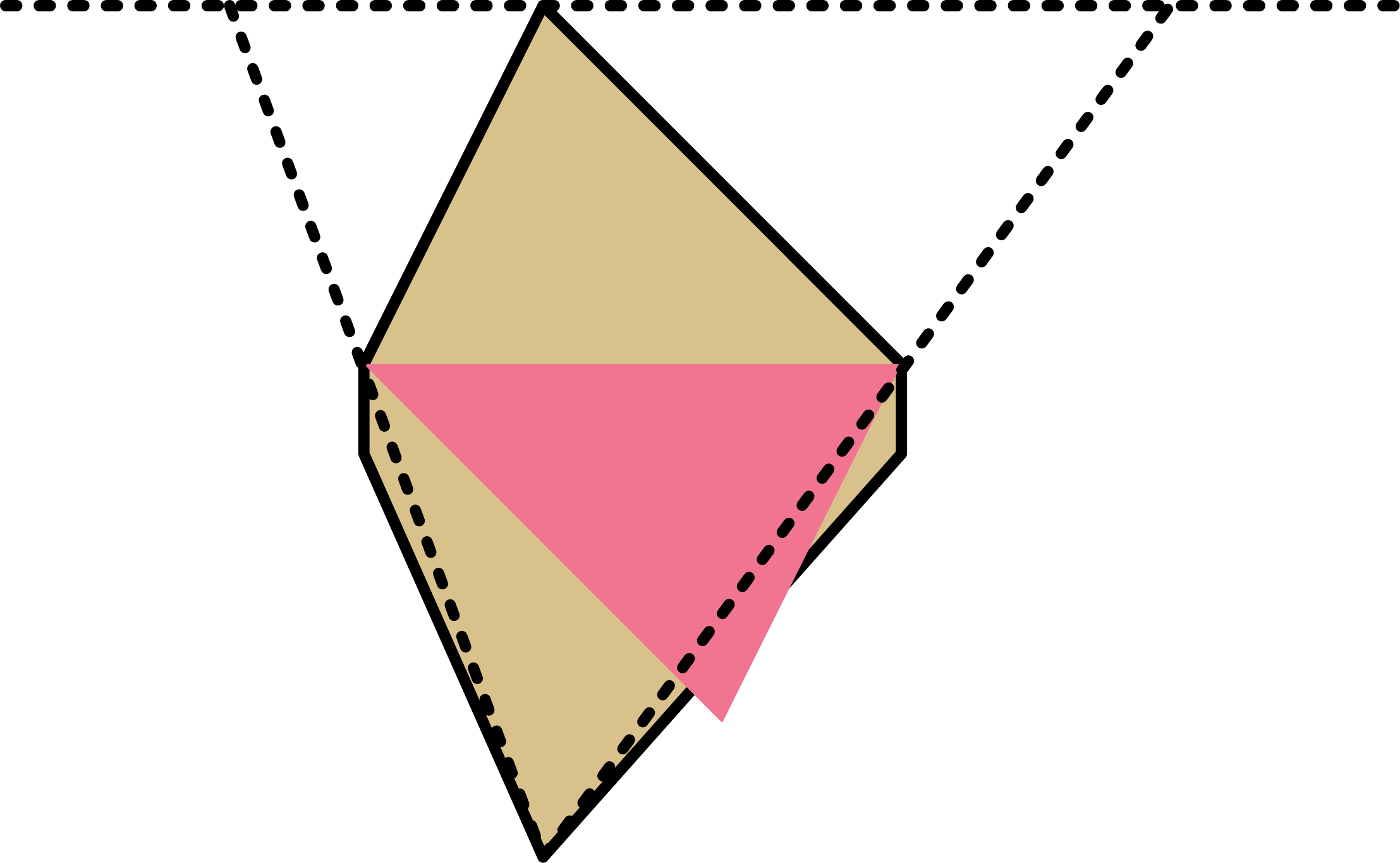}
    \put(23,35){\makebox(0,0)[c]{\strut $A$}}
    \put(1,58){\makebox(0,0)[c]{\strut $\ell$}}
    \put(39.25,57){\makebox(0,0)[c]{\strut $B$}}
    \put(53,8){\makebox(0,0)[c]{\strut $B'$}}
    \put(38.5,-2.5){\makebox(0,0)[c]{\strut $Z$}}
    \put(67,35){\makebox(0,0)[c]{\strut $C$}}
    \put(15,58){\makebox(0,0)[c]{\strut $Q_1$}}
    \put(85,58){\makebox(0,0)[c]{\strut $Q_2$}}
  \end{overpic}
  \vspace{1em}
  \caption{Bounding the distance from $B$ to $Q_1, Q_2$.}
  \label{fig:convex-projection}
\end{figure}

\begin{corollary}\label{cor:convex-projection2}
  Let $P$ be a convex polygon with five consecutive vertices $D_1, A, B, C, D_2$ (where possibly $D_1 = D_2$)
  such that the interior angle at $B$ is at least as large as the interior angles at $D_1$ and $D_2$.
  Let $\ell$ be the line through $B$ parallel to $AC$, $Q_1$ be the intersection of $D_1A$ with $\ell$, and $Q_2$ be the intersection of $D_2C$ with $\ell$.
  Then $\min\{|Q_1B|, |Q_2B|\} \le |AC|$.
\end{corollary}
\begin{proof}
  By convexity of $P$, $D_1$ is below $D_2 A$ and so $|Q_1B|$ is at most the distance from $B$ to the intersection of $D_2 A$ with $\ell$.
  The inequality thus follows from Lemma~\ref{lem:convex-projection} applied to $D_2$.
\end{proof}

\begin{proof}[Proof of Theorem~\ref{thm:convex-polygon-refolding}]
  It suffices to show that, for $n \ge 4$, any doubly covered convex $n$-gon can be reduced to a doubly covered convex $(n-1)$-gon by an $O(1)$-step refolding,
  because then we can reduce both polygons to triangles in $O(n)$ steps,
  and \cite[Theorem~2]{Refolding_CGTA} shows there is a 3-step refolding between any pair of doubly covered triangles with the same area.

  Let $P$ be a convex $n$-gon where $n \ge 4$, and let $B$ be a vertex of $P$ with the largest interior angle.
  By Corollary~\ref{cor:convex-projection2} we can label the nearby vertices of $B$ by $A, C, D$
  such that $A, B, C, D$ are consecutive and $|QB| \le |AC|$ where $Q$ is the intersection of $DC$ with $\ell$, the line through $B$ parallel to $AC$.
  Let $P'$ be the polygon obtained from $P$ by replacing $\triangle ABC$ by a rectangle with base $AC$ of the same area,
  and let $P''$ be the polygon obtained from $P$ by replacing $\triangle ABC$ by $\triangle AQC$.
  By Lemma~\ref{lem:acute}, there is an $O(1)$-step refolding between the double covers of $P$ and $P'$;
  it applies because the interior angle of $B$ is at least $90\degree$.
  Similarly, one of Lemmas~\ref{lem:acute}~or~\ref{lem:scalene} (using $|QB| \le |AC|$) shows that there is an $O(1)$-step refolding between the double covers of $P''$ and $P'$.
  Thus there is an $O(1)$-step refolding between double covers of $P$ and $P''$.
  But because $D,C,Q$ are collinear, $P''$ is a convex $(n-1)$-vertex polygon.
\end{proof}

\section{Transformation Between Tree-Shaped Polycubes}
\label{sec:polycube-transformation}

Next we consider \defn{tree-shaped $n$-cubes}, that is,
polyhedral manifolds formed from $n$ unit cubes in 3D
\defn{joined} face-to-face in a tree structure (forming a tree dual graph).
Here, when two cubes get joined together at a common face,
we remove that face from the manifold,
preserving that the manifold is homeomorphic to a sphere.
(This notion of ``join'' is a higher-dimensional analog of gluing.)
Thus every tree-shaped $n$-cube has surface area $6n-2(n-1) = 4n+2$.

We allow two cubes to be adjacent even if they are not glued together,
in which case there are two surface squares in between.
If there are no such touching cubes, we call the tree-shaped $n$-cube
\defn{slit-free}.%
\footnote{Our definition of ``slit-free'' here is less 
  restrictive than previous notions of ``well-separated''
  \cite{damian2021unfolding},
  which required at least one straight cube
  (connected to cubes on two opposite faces and nowhere else)
  between every two non-straight cubes.}
%
When the $n$-cubes are not slit-free,
we further allow multiple cubes to occupy the same location in space,
in which case we call the tree-shaped $n$-cube \defn{self-intersecting}.

All cubes of a tree-shaped $n$-cube naturally lie on a cubical grid.
Define \defn{grid cutting} to be cutting restricted to edges of the
cubical grid, and \defn{grid refolding} to be grid cutting followed by gluing
that results in another tree-shaped $n$-cube.

\begin{theorem} \label{thm:polycube}
  Any two tree-shaped $n$-cubes have an $O(n^2)$-step grid refolding,
  where all intermediate manifolds are possibly self-intersecting
  tree-shaped $n$-cubes.
  If the given tree-shaped $n$-cubes do not self-intersect and
  are slit-free, then the intermediate manifolds do not self-intersect.
\end{theorem}

To transform between two given tree-shaped polycubes $\manifold P$ and
$\manifold Q$, we mimic the ``sliding cubes'' model of
reconfiguring modular robots made up of $n$ cubes,
which was recently solved in optimal $O(n^2)$ steps \cite{abel2024universal}.
This model defines two types of operations
(see Figure~\ref{fig:polycube-transform-intro}):
\begin{enumerate}
\item \defn{Slide} a cube along a flat surface of neighboring cubes by 1 unit.
\item \defn{Rotate} a cube around the edge of an adjacent cube.
\end{enumerate}

\begin{figure}
	\centering
	\includegraphics[scale=0.3]{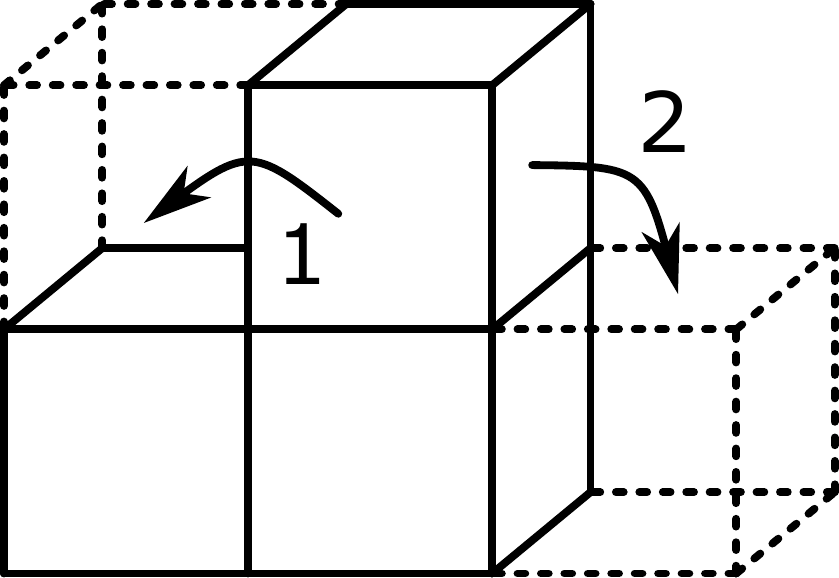}
	\caption{Two different ways an individual cube can move on a surface of a polycube: (1) sliding and (2) rotating.}
	\label{fig:polycube-transform-intro}
\end{figure}

We will show how to perform slide and rotate operations for a \defn{leaf} cube,
that is, a leaf of the dual tree in a tree-shaped polycube.
In this case, sliding can be viewed as moving a leaf cube to a new parent,
and rotating can be viewed as the leaf cube attaching to a different location
of the same parent.

To slide a leaf cube,
we perform the following refolding step,
illustrated in Figure~\ref{fig:polycube-transform-slide}:
\begin{enumerate}
\item Cut $AB$, $BE$, $ED$, $FG$, $GJ$, and $IJ$.
  These cuts free up the leaf cube to move into the adjacent location,
  as drawn in the intermediary figure in
  Figure~\ref{fig:polycube-transform-slide}.
\item Glue $AB'$ to $E'B'$, $FG'$ to $J'G'$, $E'D$ to $AB$,
  $J'I$ to $FG$, $DE$ to $BE$, and $IJ$ to $GJ$.
\end{enumerate}

\begin{figure}
	\centering
	\includegraphics[width = 0.9\linewidth]{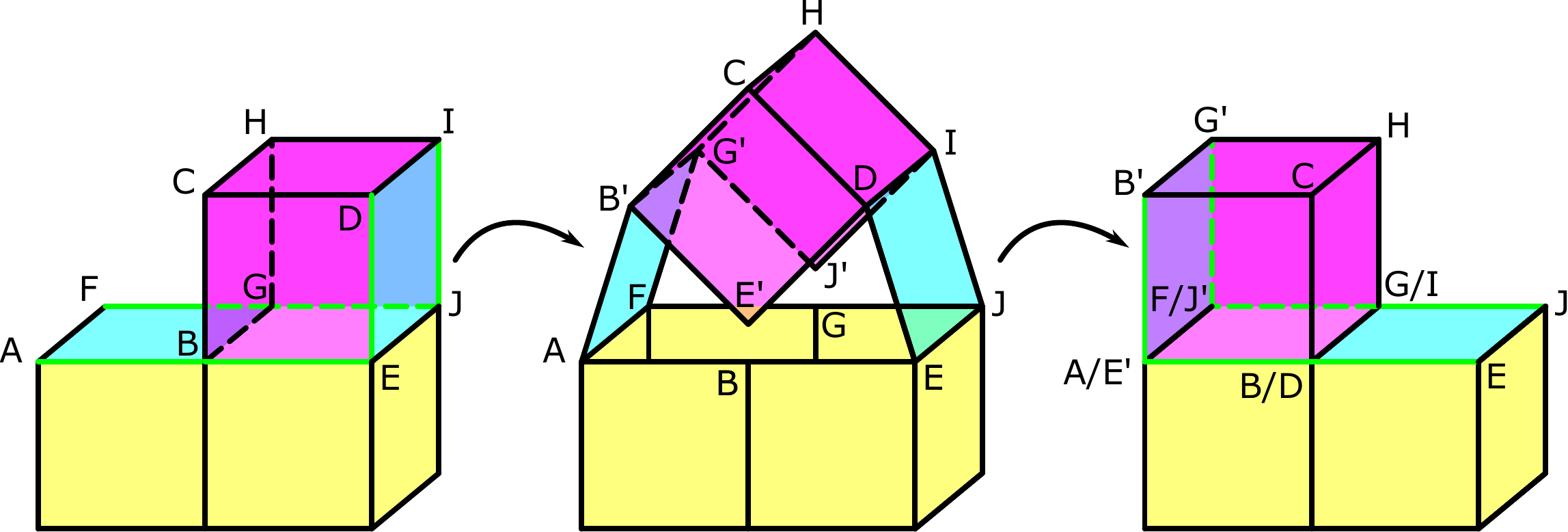}
	\caption{Sliding a leaf cube.}
	\label{fig:polycube-transform-slide}
\end{figure}

Figure~\ref{fig:polycube-transform-slide-2} shows an extension of sliding.
Here the leaf cube $IDEJCHGB$ does not move,
but it changes its parent from the cube attached below to the cube
attached on its left, effectively traversing the reflex corner.
The same refolding step as sliding applies in this case.

\begin{figure}
  \centering
  \includegraphics[width = 0.9\linewidth]{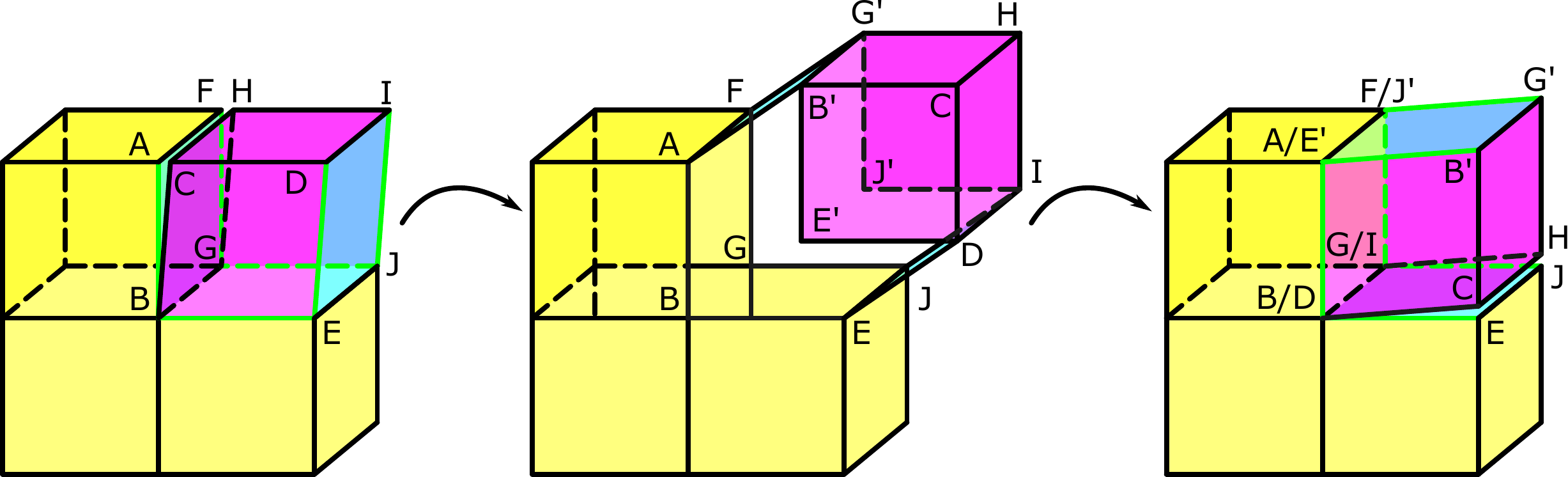}
  \caption{"Sliding" a leaf cube in a reflex corner.}
  \label{fig:polycube-transform-slide-2}
\end{figure}

To rotate a leaf cube around an edge
we perform the following refolding step,
illustrated in Figure~\ref{fig:polycube-transform-rotate}:
\begin{enumerate}
\item Cut $BA$, $AD$, $DE$, $GF$, $FI$, and $IJ$.
  Similar to the sliding procedure, these cuts free up the leaf cube to move,
  as shown in the intermediary figure in
  Figure~\ref{fig:polycube-transform-rotate}.
\item Glue $AB$ to $AD$, $FG$ to $FI$, $BA'$ to $DE$,
  $GF'$ to $IJ$, $A'D'$ to $ED'$, and $F'I'$ to $JI'$.
\end{enumerate}

\begin{figure}
  \centering
  \includegraphics[width = 0.9\linewidth]{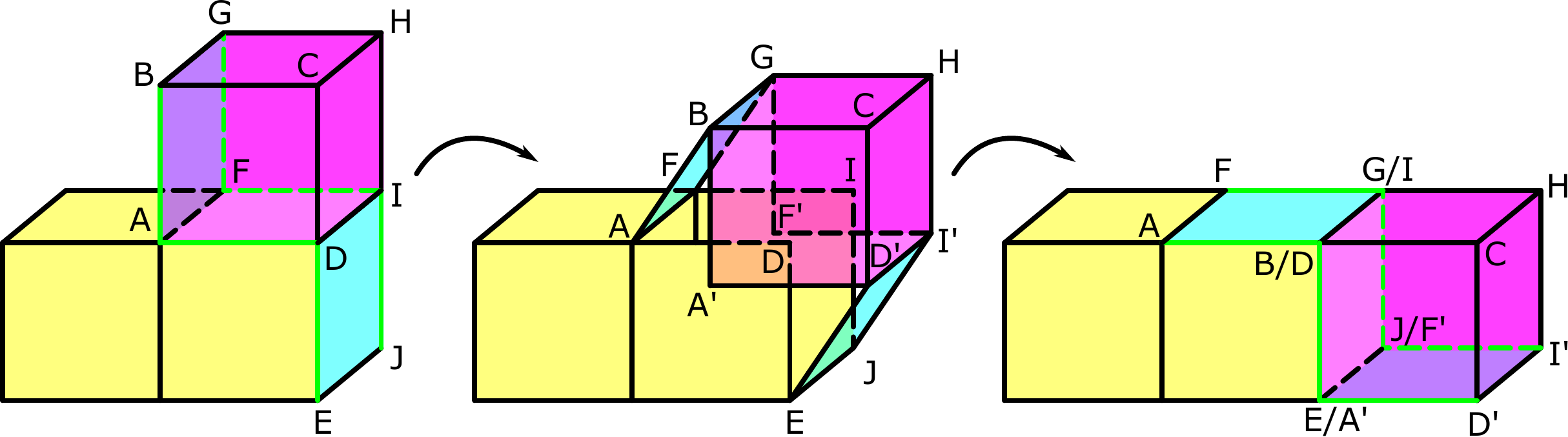}
  \caption{Rotating a leaf cube over the edge of a polycube.}
  \label{fig:polycube-transform-rotate}
\end{figure}







By combining these three operations, we can follow a simple algorithm for
transforming a given $n$-cube $\manifold P$ into a $1 \times 1 \times n$ line:
\begin{enumerate}
\item Fix one leaf cube as the \defn{root} cube~$c_1$.
  Assume by symmetry that the root cube's unique neighbor
  is in the down direction.
\item For $i = 2, 3, \dots, n$:
  \begin{enumerate}
  \item Assume $c_1, \dots, c_{i-1}$ have been arranged into an upward line,
    with $c_i$ being a current leaf.
  \item Take a leaf cube $c_i$ that is not $c_{i-1}$
    (given that there are always at least two leaves).
  \item Slide and rotate $c_i$ around the boundary of the rest of the tree
    until it reaches the root cube $c_1$, and then slide it up the line
    to place it immediately above $c_{i-1}$.
  \end{enumerate}
\end{enumerate}

This algorithm requires $O(n^2)$ steps: potentially each of the $n$ cubes
needs to traverse the surface area of the tree-shaped $n$-cube, which is $O(n)$.
It may also cause self-intersection, because it blindly follows the surface of
the tree-shaped $n$-cube, so it may place the moving leaf cube on top of
an adjacent cube in the case of touching cubes.
If the tree-shaped $n$-cube is slit-free, though,
then this simple algorithm avoids self-intersection.

To transform between two tree-shaped $n$-cubes $\manifold P$ and $\manifold Q$,
we apply the algorithm above separately
to each of $\manifold P$ and $\manifold Q$,
perform the refolding steps on $\manifold P$ to transform it into a line,
and then perform the reverse refolding steps on $\manifold Q$ to transform
the line into $\manifold Q$.  (Note that each refolding step is reversible.)
Thus we have proved Theorem~\ref{thm:polycube}.

It is tempting to apply the (much more complicated)
$O(n^2)$-step algorithm of Abel, Akitaya, Kominers, Korman, and Stock
\cite{abel2024universal}, which has the advantage of avoiding self-intersection
without any assumption of slit-freeness.
Unfortunately, sliding and rotating nonleaf cubes seem more difficult.
One approach is to transform one spanning tree into another
(probably increasing the number of steps),
but it is not even clear whether this can be accomplished by
leaf reparenting operations.

It also seems likely that some of these moves can be done in parallel
in the same refolding step, leading to fewer refolding steps.
Some models of modular robotics have parallel reconfiguration algorithms
that move a linear number of robots in each round
\cite{aloupis2008reconfiguration,aloupis2009linear}.
It remains open whether we can get similarly good bounds
in the cube sliding model or the leaf-focused sliding-by-refolding model.

\section{Conclusion}

In this paper, we showed a transformation algorithm between any two manifolds with two cut-and-glue refolding steps. When transforming between manifold~$\manifold{P}$ and manifold~$\manifold{Q}$, we go through an intermediate embeddable polyhedron which is not necessarily convex. We also showed two simpler refolding algorithms for doubly covered polyhedra and tree-shaped polycubes.

Many open questions remain:
\begin{itemize}
\item Are there examples where 1-step refolding is impossible? \cite{CommonUnfolding_CCCG2022}
\item If the two given polyhedra are convex, is there a finite-step refolding where the intermediate polyhedra are also convex? \cite[Section 25.8.3]{Demaine-O'Rourke-2007}
\item Can we extend our polycube result to avoid self-intersection without assuming slit-freeness, or to support non-tree-shaped polycubes of the same surface area?
\item Can we improve the number of refolding steps needed for the doubly covered polygon or polycube refolding algorithms?
\end{itemize}

\section*{Acknowledgments}

This work was initiated during an MIT class on Geometric Folding Algorithms
(6.849, Fall 2020).  We thank the other participants of that class ---
in particular, Josh Brunner, Hayashi Layers, Rebecca Lin, and Naveen Venkat
--- for helpful discussions and providing a productive research environment.
We thank the anonymous reviewers and Joseph O'Rourke for helpful comments
on the paper, in particular for catching a bug in an earlier version of
our manifold construction.

\bibliography{origami.bib}
\bibliographystyle{alpha}

\end{document}